\documentclass[11pt]{article}

\usepackage{times}
\usepackage{fullpage}
\usepackage{amsfonts,amssymb,amsmath,amsthm}
\usepackage{latexsym}
\def\01{\{0,1\}}

\newcommand{\eps}{\varepsilon}
\newcommand{\ket}[1]{|#1\rangle}

\newcommand{\ketbra}[2]{|#1\rangle\langle#2|}
\newcommand{\inp}[2]{\langle{#1}|{#2}\rangle} % inproduct, < | >
\newcommand{\inpc}[2]{\langle{#1},{#2}\rangle} % inproduct, < , >
\newcommand{\norm}[1]{{\left\|{#1}\right\|}}
\newcommand{\Var}{\mbox{Var}}
\newcommand{\sgn}{\mbox{sgn}}
\newcommand{\Inf}{\mbox{Inf}}
\newcommand{\E}{\mathbb{E}}
\newcommand{\Tr}{\mbox{\rm Tr}}

\newtheorem{theorem}{Theorem}
\newtheorem{lemma}{Lemma}

\begin{document}

\title{How Low Can Approximate Degree and Quantum Query Complexity be for Total Boolean Functions?\thanks{Supported
by the European Commission under the project QCS (Grant No.~255961).}}
\author{Andris Ambainis\thanks{University of Latvia, Riga. {\tt ambainis@lu.lv}}
\and
Ronald de Wolf\thanks{CWI and University of Amsterdam. {\tt rdewolf@cwi.nl}. Supported by a Vidi grant from the Netherlands Organization for Scientific Research (NWO).}
}
\date{}
\maketitle
\thispagestyle{empty}

\begin{abstract}
It has long been known that any Boolean function that depends on $n$ input variables
has both \emph{degree} and \emph{exact quantum query complexity} of $\Omega(\log n)$, and that this bound is
achieved for some functions.
In this paper we study the case of \emph{approximate degree} and \emph{bounded-error} quantum query complexity.
We show that for these measures the correct lower bound is $\Omega(\log n / \log\log n)$,
and we exhibit quantum algorithms for two functions where this bound is achieved.
\end{abstract}

\section{Introduction}

\subsection{Degree of Boolean functions}

The relations between Boolean functions and their representation as polynomials over various fields
have long been studied and applied in areas like circuit complexity~\cite{beigel:poly}, 
decision tree complexity~\cite{nisan&szegedy:degree,buhrman&wolf:dectreesurvey}, 
communication complexity~\cite{buhrman&wolf:qcclower,sherstov:dualpoly}, and many others.
In a seminal paper, Nisan and Szegedy~\cite{nisan&szegedy:degree} made a systematic
study of the representation and approximation of Boolean functions by real polynomials,
focusing in particular on the \emph{degree} of such polynomials.  
To state their and then our results, let us introduce some notation.
\begin{itemize}
\item Every function $f:\01^n\rightarrow\mathbb{R}$ has a unique representation as an $n$-variate multilinear polynomial over the reals, i.e., there exist real coefficients $a_S$ such that $f=\sum_{S\subseteq[n]} a_S\prod_{i\in S}x_i$.
Its \emph{degree} is the number of variables in a largest monomial: $\deg(f):=\max\{|S| : a_S\neq 0\}$.
\item We say $g$ \emph{$\eps$-approximates} $f$ if $|f(x)-g(x)|\leq\eps$ for all $x\in\01^n$.
%The \emph{$\eps$-approximate degree} of $f$, denoted $\deg_\eps(f)$, is $\min\{\deg(g) : g\approx_\eps f\}$.
The \emph{approximate degree} of $f$ is $\widetilde{\deg}(f):=\min\{\deg(g) : g \mbox{ $1/3$-approximates }f\}$.
\item For $x\in\01^n$ and $i\in [n]$, $x^i$ is the input obtained from $x$ by flipping the bit~$x_i$.
A variable $x_i$ is called \emph{sensitive} or \emph{influential} on $x$ (for $f$) if $f(x)\neq f(x^i)$.
In this case we also say $f$ \emph{depends} on $x_i$.
The \emph{influence} of $x_i$ (on \emph{Boolean} function $f$) is the fraction of inputs $x\in\01^n$ 
where $i$ is influential: $\Inf_i(f):=\Pr_x[f(x)\neq f(x^i)]$.
\item The \emph{sensitivity} $s(f,x)$ of $f$ at input $x$ is the number of variables that are influential on $x$,
and the sensitivity of $f$ is $s(f):=\max_{x\in\01^n}s(f,x)$.
\end{itemize}
One of the main results of~\cite{nisan&szegedy:degree} is that every function $f:\01^n\rightarrow\01$ that
depends on all $n$ variables has degree $\deg(f)\geq\log n-O(\log\log n)$
(our logarithms are to base~2). Their proof goes as follows.
On the one hand, the function $f_i(x):=f(x)-f(x^i)$ is a polynomial of degree 
at most $\deg(f)$ that is not identically equal to~0.  Hence by a version of the Schwartz-Zippel 
lemma, $f_i$ is nonzero on at least a $2^{-\deg(f)}$-fraction of the Boolean cube. 
Since $f_i(x)\neq 0$ iff $i$ is sensitive on $x$, this shows
\begin{equation}\label{eq:infdeg}
\Inf_i(f)\geq 2^{-\deg(f)}\mbox{~for every influential $x_i$.}
\end{equation}
On the other hand, with a bit of Fourier analysis (see Section~\ref{secprelim}) one can show
\[
\sum_{i=1}^n\Inf_i(f)\leq \deg(f)
\]
and hence 
\begin{equation}\label{eq:suminfdeg}
\mbox{there is an influential $x_i$ with~}\Inf_i(f)\leq \deg(f)/n.
\end{equation}
Combining (\ref{eq:infdeg}) and~(\ref{eq:suminfdeg}) implies $\deg(f)\geq \log n - O(\log\log n)$.
As Nisan and Szegedy observe, this lower bound is tight up to the $O(\log\log n)$ term for the
\emph{address function}: let $k$ be some power of~2, $n=k+\log k$, 
and view the last $\log k$ bits of the $n$-bit input as an address in the first $k$ bits. 
Define $f(x)$ as the value of the addressed variable. 
This function depends on all $n$ variables and has degree $\log k+1\leq\log n + 1$, 
because we can write it as a sum over all $\log k$-bit addresses, multiplied by the addressed variable.

\subsection{Approximate degree of Boolean functions}

Our focus in this paper is on what happens if instead of considering \emph{representation} by polynomials we consider \emph{approximation} by polynomials.  While Nisan and Szegedy studied some properties of approximate degree in their paper, they did not state a general lower bound for all functions depending on $n$ variables.
Can we modify their proof to work for approximating polynomials?
While (\ref{eq:suminfdeg}) still holds if we replace the right-hand side
by approximate degree, (\ref{eq:infdeg}) becomes much weaker.  
Since it is known that $\Inf_i(f)\geq 2^{-2s(f)+1}$~\cite[p.~443]{simon:pram} and 
$s(f)=O(\widetilde{\deg}(f)^2)$~\cite{nisan&szegedy:degree}, we have
\begin{equation}\label{eq:infapproxdeg}
\Inf_i(f)\geq 2^{-O(\widetilde{\deg}(f)^2)}\mbox{~for every influential $x_i$.}
\end{equation}
This lower bound on $\Inf_i(f)$ is in fact optimal.  For example for the $n$-bit OR-function
each variable has influence $(n+1)/2^n$ and the approximate degree is $\Theta(\sqrt{n})$.
%\footnote{Another example where Eq.~(\ref{eq:infapproxdeg}) is tight is the function of Section~\ref{secupperbound}.}  
%RdW this is wrong, x_1 has influence close to 1 since the address is 1 if the address bits are not one of the codewords
Hence modifying Nisan and Szegedy's exact-degree proof will only give an $\Omega(\sqrt{\log n})$ bound on approximate degree.
Another way to prove that same bound is to use the facts that $s(f)=O(\widetilde{\deg}(f)^2)$ 
and $s(f)=\Omega(\log n)$ if $f$ depends on $n$ bits~\cite{simon:pram}.

In Section~\ref{seclowerbound} we improve this bound to $\Omega(\log n / \log\log n)$.
The proof idea is the following.  Suppose $P$ is a degree-$d$ polynomial
that approximates $f$.  First, by a bit of Fourier analysis we show that
there is a variable $x_i$ such that the function $P_i(x):=P(x)-P(x^i)$ (which has degree $\leq d$ and expectation~0)
has low variance. We then use a concentration result for low-degree polynomials to show that $P_i$ 
is close to its expectation for almost all of the inputs.
On the other hand, since $x_i$ has nonzero influence, (\ref{eq:infapproxdeg})
implies that $|P_i|$ must be close to~1 (and hence far from its expectation)
on at least a $2^{-O(d^2)}$-fraction of all inputs.  
Combining these things then yields $d=\Omega(\log n / \log\log n)$.

\subsection{Relation with quantum query complexity}

One of the main reasons that the degree and approximate degree of a Boolean function
are interesting measures, is their relation to the \emph{quantum query complexity} of that function.  We define $Q_E(f)$ and $Q_2(f)$ as the minimal query complexity
of \emph{exact} (errorless) and $1/3$-error quantum algorithms for computing~$f$,
respectively, referring to~\cite{buhrman&wolf:dectreesurvey} for precise definitions.

Beals et al.~\cite{bbcmw:polynomialsj} established the following lower bounds
on quantum query complexity in terms of degrees:
\[
Q_E(f)\geq\deg(f)/2\mbox{~~~~and~~~~}Q_2(f)\geq\widetilde{\deg}(f)/2.
\]
They also proved that classical deterministic
query complexity is at most $O(\widetilde{\deg}(f)^6)$,
improving an earlier 8th-power result of~\cite{nisan&szegedy:degree},
so this lower bound is never more than a polynomial off for total Boolean functions.
While the polynomial method sometimes gives bounds that are polynomially
weaker than the true complexity~\cite{ambainis:degreevsqueryj}, still many
tight quantum lower bounds are based on this method~\cite{aaronson&shi:collision,ksw:dpt-siam}.

Our new lower bound on approximate degree implies that $Q_2(f)=\Omega(\log n / \log\log n)$
for all total Boolean functions that depend on $n$ variables.\footnote{In contrast,
the \emph{classical} bounded-error query complexity is lower bounded by sensitivity~\cite{nisan&szegedy:degree} and hence always $\Omega(\log n)$.}
In Section~\ref{secupperbound} we construct two functions that meet this bound,
showing that $Q_2(f)$ can be $O(\log n / \log\log n)$ for a total function that depends on $n$ bits.
Since $Q_2(f)\geq\widetilde{\deg}(f)/2$, we immediately also get that $\widetilde{\deg}(f)$ 
can be $O(\log n / \log\log n)$.
Interestingly, the only way we know to construct $f$ with asymptotically minimal $\widetilde{\deg}(f)$ 
is through such quantum algorithms---this fits into the growing sequence of classical results 
proven by quantum means~\cite{drucker&wolf:qproofs}.

The idea behind our construction is to modify the address function (which achieves the smallest degree in the exact case).  
Let $n=k+m$. We use the last $m$ bits of the input to build a \emph{quantum addressing
scheme} that specifies an address in the first $k$ bits. The value of the function is
then defined to be the value of the addressed bit.
The following requirements need to be met by the addressing scheme:
\begin{itemize}
\item
There is a quantum algorithm to compute the index $i$ addressed by $y\in\01^m$, 
using $d$ queries to~$y$;
\item
For every index $i\in\{1, \ldots, k\}$, there is a string $y\in\01^m$ that addresses $i$
(so that the function depends on all of the first $k$ bits);
\item
Every string $y\in\01^m$ addresses one of $1, \ldots, k$ (so the resulting
function on $k+m$ bits is total);
\end{itemize}
In Section~\ref{secupperbound} we give two constructions of addressing schemes that address $k=d^{\Theta(d)}$ bits using
$d$ quantum queries. Each gives a total Boolean function on $n\geq d^{\Theta(d)}$ bits that 
is computable with $d+1 = O(\log n / \log \log n)$ quantum queries: 
$d$ queries for computing the address~$i$ and 1 query to retrieve the addressed bit $x_i$.% 
\footnote{It is interesting to contrast this with ``quantum oracle interrogation''~\cite{dam:oracle}.
If we just allowed any $m$-bit address then this address could be recovered using roughly $m/2$ 
quantum queries~\cite{dam:oracle}, but not less~\cite{absw:nover2}. In other words, $d$ quantum queries could recover one of roughly $2^{2d}$ possible addresses. In the addressing schemes we consider here, where different $m$-bit strings can point to the same address, $d$~quantum queries can recover one of $d^{\Theta(d)}$ possible addresses.}

To summarize, all total Boolean functions that depend on $n$ variables have approximate degree and bounded-error quantum query complexity at least $\Omega(\log n / \log\log n)$, and that lower bound is tight for some functions.

\section{Approximate degree is $\Omega(\log n/\log\log n)$ for all total $f$}\label{seclowerbound}

\subsection{Tools from Fourier analysis}\label{secprelim}

We  use the framework of Fourier analysis on the Boolean cube.  We will just introduce what we need here,
referring to~\cite{odonnell:survey,wolf:fouriersurvey} for more details and references.
In this section it will be convenient to denote bits as $+1$ and $-1$, so a Boolean function will now be $f:\{\pm 1\}^n\rightarrow\{\pm 1\}$.
Unless mentioned otherwise, expectations and probabilities below
are taken over a uniformly random $x\in\{\pm 1\}^n$.

Define the inner product between functions $f,g:\{\pm 1\}^n\rightarrow\mathbb{R}$ as
\[
\inpc{f}{g}=\frac{1}{2^n}\sum_{x\in\{\pm 1\}^n}f(x)g(x)=\E[f\cdot g].
\]
For $S\subseteq[n]$, the function $\chi_S$ is the product (parity) of the variables indexed in $S$.
These functions form an orthonormal basis for the space of all real-valued functions on the Boolean cube.
The \emph{Fourier coefficients} of $f$ are $\widehat{f}(S)=\inpc{f}{\chi_S}$, and we can write $f$ in
its Fourier decomposition
\[
f=\sum_{S\subseteq[n]}\widehat{f}(S)\chi_S.
\]
The \emph{degree} $\deg(f)$ of $f$ is $\max\{ |S| : \widehat{f}(S)\neq 0\}$.
The \emph{expectation} or \emph{average} of $f$ is $\E[f]=\widehat{f}(\emptyset)$, and its \emph{variance}
is $\Var[f]=\E[f^2]-\E[f]^2=\sum_{S\neq\emptyset}\widehat{f}(S)^2$.
The \emph{$p$-norm} of $f$ is defined as
\[
\norm{f}_p=\E[|f|^p]^{1/p}.
\]  
This is monotone non-decreasing in $p$.
For $p=2$, Parseval's identity says 
\[
\norm{f}_2^2=\sum_S\widehat{f}(S)^2.
\]
For low-degree $f$, the famous Bonami-Beckner hypercontractive inequality 
implies that higher norms cannot be \emph{much} bigger than the 2-norm.\footnote{See for example~\cite[Lecture~16, Corollary 1.3]{odonnell:notes} or~\cite[after Theorem 4.1]{wolf:fouriersurvey} for a proof, and \cite[Chapter~5]{janson:gaussian} for more background on hypercontractivity.}

\begin{theorem}\label{thbonamibeckner}
Let $f$ be a multilinear $n$-variate polynomial.
If $q\geq 2$, then
\[
\norm{f}_q\leq (q-1)^{\deg(f)/2}\norm{f}_2.
\]
\end{theorem}

The main tool we use is the following concentration result for degree-$d$ polynomials (the degree-1 case is essentially the familiar Chernoff bound).  Its derivation from Theorem~\ref{thbonamibeckner} is folklore, see for example~\cite[Section~2.2]{dfko:fouriertails} or \cite[Theorem~5.4]{odonnell:survey}. For completeness we include the proof below.

\begin{theorem}\label{thm:degreedconcentration}
Let $F$ be a multilinear $n$-variate polynomial of degree at most $d$, 
with expectation~0 and variance $\sigma^2=\norm{F}_2^2$.
For all $t\geq(2e)^{d/2}$ it holds that
\[
\Pr[|F|\geq t\sigma]\leq \exp\left(-(d/2e)\cdot t^{2/d}\right).
\]
\end{theorem}

\begin{proof}
Theorem~\ref{thbonamibeckner} implies
\[
\E[|F|^q]=\norm{F}_q^q\leq(q-1)^{dq/2}\norm{F}_2^q=(q-1)^{dq/2}\sigma^q.
\]
Using Markov's inequality gives
\[
\Pr[|F|\geq t\sigma] = \Pr[|F|^q\geq (t\sigma)^q] 
\leq \frac{\E[|F|^q]}{(t\sigma)^q}\leq \frac{(q-1)^{dq/2}\sigma^q}{(t\sigma)^q}\leq\frac{q^{dq/2}}{t^q}.
\]
Choosing $q=t^{2/d}/e$ gives the theorem (note that our assumption on~$t$ implies $q\geq 2$).
\end{proof}

\subsection{The lower bound proof}

Here we prove our main lower bound.

\begin{theorem}
Every Boolean function $f$ that depends on $n$ input bits has 
\[
\widetilde{\deg}(f)=\Omega(\log n/\log\log n).
\]
\end{theorem}

\begin{proof}
Let $P:\mathbb{R}^n\rightarrow[-1,1]$ be a $1/3$-approximating polynomial for $f$
(the assumption that the range is $[-1,1]$ rather than $[-4/3,4/3]$ is for convenience
and does not change anything significant.)
Our goal is to show that $d:=\deg(P)$ is $\Omega(\log n/\log\log n)$.
If $d>\log n/\log\log n$ then we are already done, so assume $d\leq \log n/\log\log n$.

Define $f_i$ by $f_i(x)=(f(x)-f(x^i))/2$ and similarly define $P_i$ by $P_i(x)=(P(x)-P(x^i))/2$.
Note that both $f_i$ and $P_i$ have expectation~0.
We have $f_i(x)\in\{\pm 1\}$ if $i$ is sensitive for $x$, and $f_i(x)=0$ if $i$ is not sensitive for $x$.
Similarly for $P_i$, with an error of up to $1/3$.
Note that $\widehat{P_i}(S)=\widehat{P}(S)$ if $i\in S$ and $\widehat{P_i}(S)=0$ if $i\not\in S$.
Then
$$
\sum_{i=1}^n\norm{P_i}_2^2=\sum_{i=1}^n\sum_S\widehat{P_i}(S)^2=\sum_{i=1}^n\sum_{S\ni i}\widehat{P}(S)^2=\sum_S|S|\widehat{P}(S)^2\leq d \sum_S\widehat{P}(S)^2=d\norm{P}_2^2\leq d.
$$
Hence there exists an $i\in[n]$ for which 
$$
\norm{P_i}_2^2\leq d/n.
$$
Assume $i=1$ for convenience. Because every variable (including $x_1$) is influential,
Eq.~(\ref{eq:infapproxdeg}) implies
\[
\Inf_1(f)\geq 2^{-O(d^2)}.
\]
Define $\sigma^2=\Var[P_1]=\norm{P_1}_2^2\leq d/n$.
Set $t=1/2\sigma\geq \sqrt{n/4d}$. Then $t\geq(2e)^{d/2}$ for sufficiently large $n$,
because we assumed $d\leq \log n / \log\log n$.
Now use Theorem~\ref{thm:degreedconcentration} to get
\begin{align*}
\Inf_1(f) & = \Pr[f_1(x)\in\{\pm 1\}]\\
           & = \Pr[|P_1(x)|\geq 1/2]\\
           & = \Pr[|P_1(x)|\geq t\sigma]\\
           & \leq  \exp\left(-(d/2e)\cdot t^{2/d}\right)\\
           & \leq  \exp\left(-(d/2e)\cdot (n/4d)^{1/d}\right).
\end{align*}
Combining the upper and lower bounds on $\Inf_1(f)$ gives
\[
2^{-O(d^2)} \leq \exp\left(-(d/2e)(n/4d)^{1/d}\right).
\]
Taking logarithms of left and right-hand side and negating gives
\[
O(d^2) \geq (d/2e)(n/4d)^{1/d}.
\]
Dividing by $d$ and using our assumption that $d\leq \log n / \log\log n$ implies, for sufficiently large~$n$:
\[
\log n \geq (n/4d)^{1/d}.
\]
Taking logarithms once more we get 
\[
d\geq\log(n/4d)/\log\log n=\log n / \log\log n - O(1),
\]
which proves the theorem.
\end{proof}

Note that the constant factor in the $\Omega(\cdot)$ is essentially~1 for any constant approximation error.
The $\Omega(\log n / \log\log n)$ bound remains valid even for quite large errors:
the same proof shows that for every constant $\gamma<1/2$, every polynomial $P$ for which $\sgn(P(x))=f(x)$ 
and $|P(x)|\in[1/n^\gamma,1]$ for all $x\in\{\pm 1\}^n$, has degree $\Omega(\log n / \log\log n)$.
This lower bound no longer holds if $\gamma=1$; for example for odd~$n$, the degree-1 polynomial 
$\sum_{i=1}^n x_i/n$ has the same sign as the majority function, and $|P(x)|\in [1/n,1]$ everywhere.

\section{Functions with quantum query complexity $O(\log n / \log\log n)$}\label{secupperbound}

In this section we exhibit two $n$-bit Boolean functions whose bounded-error quantum query complexity 
(and hence approximate degree) is $O(\log n / \log\log n)$.

\begin{theorem}
There is a Boolean function $f:\01^n\rightarrow\01$ that depends on all $n$ variables and has 
\[
Q_2(f) = O\left(\frac{\log n}{\log \log n}\right) .
\]
\end{theorem}

\begin{proof}
Let us call a function $a(x_1, \ldots, x_m)$ of $m$ variables $x_1,\ldots,x_m\in\01$ a \emph{$k$-addressing scheme} if 
$a(x_1,\ldots,x_m)\in[k]$ and, for every $i\in[k]$, there exist $x_1,\ldots,x_m\in\01$ such that $a(x_1,\ldots,x_m)=i$. 

\begin{lemma}
\label{lem:address}
For every $t>0$, there exists a $k$-addressing scheme $a(x_1, \ldots, x_m)$
with $k=t^t$ that can be computed with error probability $\leq 1/3$ using $O(t)$ quantum queries.
\end{lemma}

\begin{proof}
In Sections~\ref{sec:add1} and~\ref{sec:add2} we give two constructions of addressing schemes achieving this bound.
\end{proof}

Set $m=t^2$, $k=t^t$, and $n=m+k$.
Without loss of generality, we assume all variables $x_1,\ldots,x_m$
in the $k$-addressing scheme $a(x_1, \ldots, x_m)$ from Lemma~\ref{lem:address} are significant. 
(Otherwise remove the insignificant variables and decrease $m$.)
Define the following $n$-bit Boolean function:
\[ 
f(x_1,\ldots,x_n) = x_{a(x_{k+1},x_{k+2},\ldots,x_{k+m})} .
\]
Then $f$ can be computed with $O(t)+1$ queries and the number of variables is $n>k=t^t$. Hence,
\[ 
\frac{\log n}{\log \log n} \geq \frac{t \log t}{\log t + \log \log t} = (1+o(1)) t.
\] 
\end{proof}

\subsection{Addressing scheme: 1st construction}
\label{sec:add1}

Define the scheme in the following way. 
We select $k=t^t$ words $w^{(i)}$ of $m$ bits each, such that any two distinct words $w^{(i)}$ and $w^{(j)}$ have Hamming distance in the interval $I=[\frac{m}{2} - c t \sqrt{t \log t},\frac{m}{2} + c t \sqrt{t \log t}]$.

One can for example show the existence of such strings using a standard application of the probabilistic method, as follows. Select the $w^{(i)}$ randomly from $\01^m$. 
For distinct $i$ and $j$, the expected Hamming distance between $w^{(i)}$ and $w^{(j)}$ equals $m/2$.  By a Chernoff bound, the probability that this Hamming distance is outside of the interval $I$ is $2^{-\Omega(c^2t^3\log(t)/m)}=2^{-\Omega(c^2t\log t)}$.
If we choose $c$ a sufficiently large constant then this probability is $o(1/{k\choose 2})$.  Since there are ${k\choose 2}$ distinct $i,j$-pairs, the union bound implies that with probability $1-o(1)$, all pairs of words $w^{(i)}$ and $w^{(j)}$ have Hamming distance in the interval~$I$.

For input $x\in\01^m$, define $a(x):=i$ if $x=w^{(i)}$, and $a(x):=1$ if $x$ does not equal any of $w^{(1)},\ldots,w^{(k)}$.
We select $t'=O(t)$ so that
\[ 
\left( \frac{2 c\sqrt{\log t}}{\sqrt{t}} \right)^{t'} \leq \frac{1}{t^{2t}} .
\]
Let
\[ 
\ket{\psi} = \frac{1}{\sqrt{m}} \sum_{j=1}^m(-1)^{x_j} \ket{j} .
\]
Let $\ket{\psi_i}$ be the state $\ket{\psi}$ defined above
if $x=w^{(i)}$. If $i\neq j$, we have 
\[ 
\inp{\psi_i^{\otimes t'}}{\psi_j^{\otimes t'}} = 
\left( \inp{\psi_i}{\psi_j}\right)^{t'} 
\leq \left( \frac{2 c\sqrt{\log t}}{\sqrt{t}} \right)^{t'} \leq \frac{1}{t^{2t}} .
\]

The following lemma is quantum computing folklore.  For the sake of completeness we include a proof in Appendix~\ref{app:proofofqpovm}.

\begin{lemma}
\label{lem:dist}
Let $k\geq 1$ and $\ket{\phi_1},\ldots,\ket{\phi_k}$ be states such that 
$|\inp{\phi_i}{\phi_j}|\leq 1/k^2$ whenever $i\neq j$.
Then there is a measurement that, given $\ket{\phi_i}$, 
produces outcome~$i$ with probability at least $2/3$.
\end{lemma}

\noindent
We will apply this lemma to the $k$ states $\ket{\phi_i} = \ket{\psi_i}^{\otimes t'}$.
Our $O(t)$-query quantum algorithm is as follows:
\begin{enumerate}
\item
Use $t'=O(t)$ queries to generate $\ket{\psi}^{\otimes t'}$.
\item
Apply the measurement of Lemma~\ref{lem:dist}.
\item
If the measurement gives some $i\neq 1$, then use Grover's search algorithm~\cite{grover:search,bhmt:countingj} 
(with error probability $\leq 1/3$) to search for $j\in[m]$ such that $x_j \neq w^{(i)}_j$.
\item
If no such $j$ is found, then output~$i$. Else output~1.
\end{enumerate}
The number of queries is $O(t)$ to generate $\ket{\psi}^{\otimes t'}$ and $O(\sqrt{m})=O(t)$ for Grover search, so $O(t)$ in total.

If the input $x$ equals some $w^{(i)}$,
then the measurement of Lemma~\ref{lem:dist} will produce the correct~$i$ with
probability at least $2/3$ and Grover search will not find $j$ s.t.~$x_j \neq w^{(i)}_j$.
Hence, the whole algorithm will output~$i$ with probability at least $2/3$.
If the input $x$ is not equal to any $w^{(i)}$,
then the measurement will produce some $i$ but Grover search will find $j$ s.t.~$x_j \neq w^{(i)}_j$, with probability at least $2/3$. 
As a result, the algorithm will output the correct answer~1 with probability at least $2/3$ in this case.

\subsection{Addressing scheme: 2nd construction}
\label{sec:add2}

Our second addressing scheme is based on the Bernstein-Vazirani algorithm~\cite{bernstein&vazirani:qcomplexity}.
For a string $z\in\01^s$, let $h(z)$ be its $2^s$-bit Hadamard codeword: $h(z)_j=z\cdot j$ mod~2, where $j$ ranges over all indices $\in\01^s$,
and $z\cdot j$ denotes the inner product of the two $s$-bit strings $z$ and $j$. 
The Bernstein-Vazirani algorithm recovers $z$ with probability~1 using only one quantum query if its $2^s$-bit input is of the form $h(z)$.
For our addressing scheme, we set $s=\log\log k-\log\log\log k$ and $t=(\log k)/s$ (assume for simplicity these numbers are integers).
Note that $k=t^{(1+o(1))t}$.
The $m$-bit input $x$ to the addressing scheme consists of $t$ blocks $x^{(1)},\ldots,x^{(t)}$ of $2^s$ bits each, so $m=t2^s=O(t^2)$.
Define the addressing scheme as follows:
\begin{quote}
If $x$ is of the form $h(z^{(1)})\ldots h(z^{(t)})$ then set $a(x):=z^{(1)}\ldots z^{(t)}$.
Otherwise set $a(x):=0^{\log k}$.
\end{quote}
Note that the value of $a(x)$ is a $\log k$-bit string, and that the function is surjective. 
Hence, identifying $\01^{\log k}$ with $[k]$, the function~$a$ addresses a space of $k$ bits.

The following algorithm computes $a(x)$ with $O(t)$ quantum queries:
\begin{enumerate}
\item Use the Bernstein-Vazirani algorithm $t$ times, once on each $x^{(j)}$, computing $z^{(1)},\ldots,z^{(t)}\in\01^s$.
\item Use Grover~\cite{grover:search,bhmt:countingj} to check if $x=x^{(1)}\ldots x^{(t)}$ equals the $m$-bit string $h(z^{(1)})\ldots h(z^{(t)})$.
\item If yes, output $a(x)=z^{(1)}\ldots z^{(t)}$. Else output~$0^{\log k}$.
\end{enumerate}
The query complexity is $t$~queries for the first step and $O(\sqrt{m})=O(t)$ for the second.

If the input $x$ is the concatenation of $t$ Hadamard codewords $h(z^{(1)}),\ldots,h(z^{(t)})$, 
then the first step will identify the correct $z^{(1)},\ldots,z^{(t)}$ with probability~1, and the second step will not find any discrepancy.
On the other hand, if the input is not the concatenation of $t$ Hadamard codewords then the two strings compared in step~2 are not equal,
and Grover search will find a discrepancy with probability at least $2/3$, in which case the algorithm outputs the correct value $0^{\log k}$.

\section{Conclusion}

We gave an optimal answer to the question how low approximate degree and bounded-error quantum query complexity can be for total Boolean functions depending on $n$ bits.
We proved a general lower bound of $\Omega(\log n / \log\log n)$,
and exhibited two functions where this bound is achieved.
The latter upper bounds are obtained by variations of the address function that are suitable for quantum algorithms.

\paragraph{Acknowledgements.}
Eq.~(\ref{eq:infapproxdeg}) was observed in email discussion between RdW and Scott Aaronson in 2008.
We thank Art\={u}rs Ba\v{c}kurs, Oded Regev, Mario Szegedy, and the anonymous CCC referees for useful comments and references.

\bibliographystyle{alpha}
%\bibliography{qc}

\newcommand{\etalchar}[1]{$^{#1}$}

\appendix

\section{Proof of Lemma~\ref{lem:dist}}\label{app:proofofqpovm}

The lemma is obvious for $k=1$, so we can assume $k\geq 2$.
Let Hilbert space $\cal H$ be the span of the states $\ket{\phi_1},\ldots,\ket{\phi_k}$, and define $A=\sum_{i=1}^k\ketbra{\phi_i}{\phi_i}$ as an operator on this space. We want to show that $A$ is close to the identity operator on~$\cal H$. 
We first show that $A\ket{\phi_j}$ is close to $\ket{\phi_j}$ for all $j\in[k]$.
Define $\ket{\delta_j}=A\ket{\phi_j}-\ket{\phi_j}$.
We have 
$$
\norm{\delta_j}=\norm{\sum_{i\in[k]\backslash\{j\}}\ketbra{\phi_i}{\phi_i}\ket{\phi_j}}\leq\sum_{i\in[k]\backslash\{j\}}|\inp{\phi_i}{\phi_j}|\leq \frac{k-1}{k^2}.
$$
Now we show $A\ket{v}$ is close to $\ket{v}$ for an arbitrary unit vector $\ket{v}=\sum_{j=1}^k\alpha_j\ket{\phi_j}$ in~$\cal H$.
Define $a:=\sum_{j=1}^k|\alpha_j|^2$. We have
$$
1=\inp{v}{v}=\sum_{i,j=1}^k\alpha^*_i\alpha_j\inp{\phi_i}{\phi_j}=a + \sum_{i\neq j}\alpha^*_i\alpha_j\inp{\phi_i}{\phi_j}.
$$
Also, using the Cauchy-Schwarz inequality,
$$
\sum_{i\neq j}\alpha^*_i\alpha_j\inp{\phi_i}{\phi_j}
\leq \sqrt{\sum_{i\neq j}|\alpha_i|^2|\alpha_j|^2}\sqrt{\sum_{i\neq j}|\inp{\phi_i}{\phi_j}|^2}
\leq \sqrt{\sum_{i,j}|\alpha_i|^2|\alpha_j|^2}\sqrt{\sum_{i,j} 1/k^4}
=a/k.
$$
This implies $1\geq a-a/k$ and hence $a\leq 1/(1-1/k)=k/(k-1)$.
We have 
$$
A\ket{v}=\sum_{j=1}^k \alpha_jA\ket{\phi_j}=\sum_{j=1}^k \alpha_j(\ket{\phi_j}+\ket{\delta_j})=\ket{v}+\sum_{j=1}^k \alpha_j\ket{\delta_j}.
$$
This implies, again using Cauchy-Schwarz,
$$
\norm{A\ket{v}-\ket{v}}\leq \sum_{j=1}^k \alpha_j\norm{\delta_j}\leq \sqrt{\sum_{j=1}^k|\alpha_j|^2}\sqrt{\sum_{j=1}^k\norm{\delta_j}^2}\leq \sqrt{\frac{k}{k-1}}\sqrt{\frac{k(k-1)^2}{k^4}}=\sqrt{\frac{k-1}{k^2}}\leq \frac{1}{2}.
$$
Hence $A\leq\frac{3}{2}I$.

Our measurement will consist of the operators $E_i=\frac{2}{3}\ketbra{\phi_i}{\phi_i}$ for all $i\in [k]$, and $E_0=I-\sum_{i=1}^k E_i$. By the previous discussion $E_0=I-\frac{2}{3}A\geq 0$, so this is a well-defined measurement (more precisely, a POVM).
Given state $\ket{\phi_i}$, $i\in[k]$, the probability that our measurement produces the correct outcome~$i$ equals
$\Tr(E_i\ketbra{\phi_i}{\phi_i})=2/3$.

\end{document}